\documentclass{article}[11pt]
\usepackage{amsthm}
\usepackage{amssymb}
\usepackage{amsthm}
\usepackage{multicol}
\usepackage{hyperref}
\usepackage{amsfonts}
\usepackage[numbers]{natbib}
\usepackage[left=1in,top=1in,right=1in,bottom=1in]{geometry}
\usepackage{authblk}
\usepackage{amsmath}
\usepackage{xcolor}  %
\newcommand{\work}{W}
\newcommand{\spn}{S}
\newcommand{\worktrans}{H}
\newcommand{\probsize}{x}
\newcommand{\probsizeB}{y}
\newcommand{\prob}{z}
\newcommand{\probB}{y}
\newcommand{\probDom}{I}
\newcommand{\tollwork}{a}
\newcommand{\superadd}{g_1}
\newcommand{\tolldiv}{g_2}
\newcommand{\tollspan}{b}
\newcommand{\size}{size}
\newcommand{\sub}{h}
\newcommand{\expected}[1]{E\left[#1\right]}
\newcommand{\meanbound}{m}
\newcommand{\detrec}{u}
\newcommand{\detrecinv}{u'}
\newcommand{\probgt}[2]{Prob\left[#1 > #2\right]}
\newcommand{\probeq}[2]{Prob\left[#1 = #2\right]}
\newcommand{\probge}[2]{Prob\left[#1 \geq #2\right]}
\newcommand{\boundfun}{D}
\newcommand{\recterm}{d}
\newcommand*{\Sref}[1]{\hyperref[#1]{\S\ref*{#1}}}
\newcommand*{\eqnref}[1]{Equation \hyperref[#1]{\ref*{#1}}}
\newcommand*{\thmref}[1]{Theorem \hyperref[#1]{\ref*{#1}}}
\newcommand*{\lemref}[1]{Lemma \hyperref[#1]{\ref*{#1}}}
\newcommand*{\Krec}[1]{K^{#1}}

\newtheorem{thm}{Theorem}
\newtheorem{lemma}{Lemma}

\definecolor{StringRed}{rgb}{.637,0.082,0.082}
\definecolor{CommentGreen}{rgb}{0.0,0.55,0.3}
\definecolor{KeywordBlue}{rgb}{0.0,0.3,0.55}
\definecolor{LinkColor}{rgb}{0.55,0.0,0.3}
\definecolor{CiteColor}{rgb}{0.55,0.0,0.3}
\definecolor{HighlightColor}{rgb}{0.0,0.0,0.0}

\definecolor{grey}{rgb}{0.5,0.5,0.5}
\definecolor{red}{rgb}{1,0,0}

\hypersetup{%
  linktocpage=true, pdfstartview=FitV,
  breaklinks=true, pageanchor=true, pdfpagemode=UseOutlines,
  plainpages=false, bookmarksnumbered, bookmarksopen=true, bookmarksopenlevel=3,
  hypertexnames=true, pdfhighlight=/O,
  colorlinks=true,linkcolor=LinkColor,citecolor=CiteColor,
  urlcolor=LinkColor
}
 
\title{Probabilistic Recurrence Relations for\\ Work and Span of Parallel Algorithms}
\author{Joseph Tassarotti}
\affil{Computer Science Department, Carnegie Mellon University}
\affil{\texttt{jtassaro@cs.cmu.edu}}
\date{}

\begin{document}
\maketitle
\begin{abstract}
  
  In this paper we present a method for obtaining tail-bounds for
  random variables satisfying certain probabilistic recurrences that
  arise in the analysis of randomized parallel divide and conquer
  algorithms. In such algorithms, some computation is initially done
  to process an input $\probsize$, which is then randomly split into
  subproblems $\sub_1(\probsize),\dots,\sub_n(\probsize)$, and the algorithm
  proceeds recursively in parallel on each subproblem. The total
  \emph{work} on input $\probsize$, $\work(\probsize)$, then satisfies a
  probabilistic recurrence of the form
  $\work(\probsize) = \tollwork(\probsize) + \sum_{i=1}^n\work(\sub_i(\probsize))$,
  and the \emph{span} (the longest chain of sequential dependencies), satisfies
  $\spn(\probsize) = \tollspan(\probsize) + \max_{i=1}^n\spn(\sub_i(\probsize))$,
  where $\tollwork(\probsize)$ and $\tollspan(\probsize)$ are the work and span 
  to split $\probsize$ and combine the results of the recursive calls.

  Karp has previously presented methods for obtaining tail-bounds in
  the case when $n = 1$, and under certain stronger assumptions for
  the work-recurrence when $n\geq 2$, but left open the question of
  the span-recurrence. We first show how to extend his
  technique to handle the span-recurrence. We then show that in some
  cases, the work-recurrence can be bounded under simpler
  assumptions than Karp's by transforming it into a related
  span-recurrence and applying our first result. We demonstrate
  our results by deriving tail bounds for the work and span of quicksort and
  the height of a randomly generated binary search tree.

\end{abstract}
 \newpage

\section{Introduction}

The analysis of sequential divide-and-conquer algorithms
involves analyzing recurrences of the form 
\begin{equation}
\label{eqn:work-rec}
\work(\probsize) = \tollwork(\probsize) + \sum_{i=1}^n\work(\sub_i(\probsize))
\end{equation}
where $\tollwork(\probsize)$ is the ``toll'' cost to initially process an input
of size $\probsize$ and split it into smaller inputs
$\sub_1(\probsize),\dots,\sub_n(\probsize)$, which are then processed
recursively. When the functions $\sub_1, \dots, \sub_n$
are deterministic, there are a number of
well-known ``cook-book'' methods that can be used to obtain asymptotic bounds in many
cases~\cite{AkraB98, CLRS, Roura01}.

In algorithms where the recursive calls are done in parallel, another
measure of cost becomes important: the \emph{span}~\cite{BlellochGr95, BlumofeL99} measures
the longest sequential dependency in the execution, and is important
for measuring the speed-up of the algorithm as the number of
processors is increased. The recurrence for the span has the form:
\begin{equation}
\label{eqn:span-rec}
\spn(\probsize) = \tollspan(\probsize) + \max_{i=1}^{n}\spn(\sub_i(\probsize))
\end{equation}
where $\tollspan(\probsize)$ is the span for processing and dividing the input, but we take the maximum of the recursive calls instead of summing
because we are interested in the longest dependency.  These kinds of
recurrences involving maxima also arise when analyzing heights of
search trees or stack space usage of sequential algorithms.
 Again, when the $\sub_i$ are deterministic,
the above recurrence often simplifies to something of the form
$\spn(\probsize) = \tollspan(\probsize) + \spn(\sub_i(\probsize))$ for some
$\sub_i$, and so the usual cook-book methods can be used to obtain asymptotic
bounds.

However, when the $\sub_i(\probsize)$ are random variables (either because
the algorithm uses randomness when dividing the input into
subproblems, or when considering average case complexity),
then $\work(\probsize)$ and $\spn(\probsize)$ are themselves random variables and the
situation is more complicated. First, when considering the expected value of
the work, we can at least use linearity of expectation to get
\begin{equation}
\expected{\work(\probsize)} = \tollwork(\probsize) + \sum_{i=1}^n\expected{\work(\sub_i(\probsize))}
\end{equation}
and a variety of techniques can be used to solve the resulting recurrence. But in the case of span,
expectation does not commute with taking a maximum, so that in general
{$\expected{\max_{i=1}^n\spn(\sub_i(\probsize))} \neq \max_{i=1}^n\expected{\spn(\sub_i(\probsize))}$},
and we merely have
{$\expected{\max_{i=1}^n\spn(\sub_i(\probsize))} \geq \max_{i=1}^n\expected{\spn(\sub_i(\probsize))}$}.
As a result, the expected value of the span does not obey a simple
recurrence.  A second source of complexity is that we may be
interested not just in obtaining bounds on expectation, but may also
wish to know that with high probability the work and span are close to
this expected value.

Although a number of methods are known for obtaining results for
certain recurrences of these forms~\cite{Devroye86, Drmota01,
  DrmotaTrees, FlajoletSedgewick}, these often rely on sophisticated
tools from probability theory or analytic combinatorics and are not
simple cook-book results that can be easily applied.

In contrast to some of the more advanced techniques mentioned above,
Karp~\cite{Karp94} has presented relatively simple methods for
obtaining tail-bounds for work recurrences of the form in
\eqnref{eqn:work-rec}:

\begin{thm}[{\cite[Theorem 1.2]{Karp94}}]
\label{thm:karp-simple}

  Suppose $n=1$ in \eqnref{eqn:work-rec}, so that $\work(\probsize) =
  \tollwork(\probsize) + \work(\sub(\probsize))$ where $\sub(x)$ is a random variable
  such that $0 \leq h(x) \leq x$.  Furthermore,
  assume $\expected{\sub(\probsize)} \leq \meanbound(\probsize)$ where
  $0 \leq \meanbound(\probsize) \leq \probsize$ and
  $\meanbound(\probsize)/\probsize$ is non-decreasing. Suppose $\detrec$
  is a solution to the recurrence $\detrec(\probsize) =
  \tollwork(\probsize) + \detrec(\meanbound(\probsize))$.
  If $\tollwork$ and $\detrec$ are continuous and $\tollwork$ is strictly increasing
  on the set $\{x\ |\ \tollwork(x) >0\}$, then for all positive integers $w$,
\[
\probgt{\work(\probsize)}{\detrec(\probsize) + w \tollwork(\probsize)} 
  \leq
 \left(\frac{\meanbound(\probsize)}{\probsize}\right)^w
\]
\end{thm}
In other words, if there is a single recursive call and we can suitably
bound the expected size of $\sub(\probsize)$, this
lets us derive tail-bounds for the work in terms of a simple
deterministic recurrence. Often $\meanbound(x)/x$ will be a constant fraction such as $1/2$ or $3/4$,
so that the recurrence for $\detrec$ can be solved straightforwardly.
A second result covers the case where there is more than one recursive call:
\begin{thm}[{\cite[Theorem 1.5]{Karp94}}]
\label{thm:karp-work-simple}

  Let $\work$ obey the recurrence in \eqnref{eqn:work-rec}.
  Suppose that for all $(\probsizeB_1,\dots,\probsizeB_n)$
  in the support of the joint distribution $(\sub_1(\probsize), \dots, \sub_n(\probsize))$, 
$\expected{\work(\probsize)} \geq \sum_{i=1}^n\expected{\work(\probsizeB_i)}$. 
Then for all positive $w$:
\[
\probge{\work(\probsize)}{(w + 1)\expected{\work(\probsize)}}
  <
  \frac{1}{e^w}
\]
\end{thm}
Both of these results are proved in the more general setting where the distributions of the $\sub_i$ (and hence $\work$) may depend on more than just the size of the input, that is we consider recurrences:
\begin{equation}
\label{eqn:work-rec-general}
\work(\prob) = \tollwork(\size(\prob)) + \sum_{i=1}^n\work(\sub_i(\prob))
\end{equation}
where $\size(\prob)$ gives the size of input $\prob$.

Despite the generality of the above results, Karp notes that the assumption
$\expected{\work(\probsize)} \geq \sum_{i=1}^n\expected{\work(\probsizeB_i)}$
in \thmref{thm:karp-work-simple} is not satisfied in some important cases. Even when it is satisfied, it may be difficult to show since we need to know more than just the asymptotic value of $\expected{\work(\probsize)}$. Moreover, Karp mentions the problem of handling the span-recurrences of the form
in \eqnref{eqn:span-rec} but leaves it open.  Although
subsequent work~\cite{ChaudhuriD97} has presented an alternative to
\thmref{thm:karp-simple} that has some weaker assumptions, to our knowledge, no
follow-up has satisfactorily extended these results to span-recurrences or weakened the
assumption needed for \thmref{thm:karp-work-simple}.

In this paper we present a generalization of \thmref{thm:karp-simple} to
handle both span and work-recurrences for the general case of $n \geq
1$, without the restrictive assumption needed in \thmref{thm:karp-work-simple}.
There are two key aspects of our results that distinguish them:

\begin{itemize}
\item Rather than having a function $\meanbound$ such that
  $\expected{\sub(\probsize)} \leq \meanbound(\probsize)$, we require
  $\expected{\max_{i=1}^n\sub_i(\probsize)} \leq
  \meanbound(\probsize)$. That is, $\meanbound(x)$ must bound the
  expectation of the \emph{maximum} of the problem sizes.  In many
  applications, bounding this is not significantly more difficult than
  the bound needed in \thmref{thm:karp-simple}.
  
\item 
  Instead of checking that for all $(\probsizeB_1, \dots, \probsizeB_n)$
  in the support of the distribution $(\sub_1(\probsize), \dots, \sub_n(\probsize))$, 
$\expected{\work(\probsize)} \geq \sum_{i=1}^n\expected{\work(\probsizeB_i)}$,
  our theorems are parameterized by functions $\superadd$ and $\tolldiv$
  for which we require
$\superadd(\probsize) \geq \sum_{i=1}^n\superadd(\probsizeB_i)$
and
$\tolldiv(\probsize) \geq \sum_{i=1}^n\tolldiv(\probsizeB_i))$.
Often, it is easier to exhibit $\superadd$ and $\tolldiv$ 
with these properties than to check that the expected value has them.

\end{itemize}
Simplified forms of our main results are:

\begin{thm}
\label{thm:span-simple}
  Suppose $\spn(\probsize) = \tollspan(\probsize) +
  \max_{i=1}^n\spn(\sub_i(\probsize))$, where for all $\probsize$,
  there exists some $k$ such that the recurrence terminates after at
  most $k$ recursive calls.
Let $\superadd$ be a monotone function such that 
for all $(\probsizeB_1,\dots,\probsizeB_n)$
in the support of the distribution $(\sub_1(\probsize), \dots, \sub_n(\probsize))$, 
$\sum \superadd(\probsizeB_i) \leq \superadd(\probsize)$.
  Furthermore, assume $\expected{\max_{i=1}^n\sub_i(\probsize)} \leq \meanbound(\probsize)$ where
  $0 \leq \meanbound(\probsize) \leq \probsize$ and
  $\meanbound(\probsize)/\probsize$ is non-decreasing. Assume $\detrec$
  is a solution to the recurrence $\detrec(\probsize) =
  \tollspan(\probsize) + \detrec(\meanbound(\probsize))$.
  If $\tollspan$ and $\detrec$ are continuous and $\tollspan$ is strictly increasing
  on the set $\{x\ |\ \tollspan(x) >0\}$, then for all positive integers $w$ and $\probsize$ such that $\superadd(\probsize) \geq 1$,
\[
\probgt{\spn(\probsize)}{\detrec(\probsize) + w \tollspan(\probsize)} 
  \leq
 \superadd(\probsize)\left(\frac{\meanbound(\probsize)}{\probsize}\right)^w
\]
\end{thm}

\begin{thm}
\label{thm:work-simple}
Suppose $\work(\probsize) = \tollwork(\probsize) +
\sum_{i=1}^n\work(\sub_i(\probsize))$, and the same assumptions about
termination of the recurrence and $m(x)$ hold as in
\thmref{thm:span-simple}.  Let $\superadd$ and $\tolldiv$ be
functions such that for all $(\probsizeB_1,\dots,\probsizeB_n)$
in the support of the distribution $(\sub_1(\probsize), \dots,
\sub_n(\probsize))$,
$\sum \superadd(\probsizeB_i) \leq \superadd(\probsize)$
and
$\sum \tolldiv(\probsizeB_i) \leq \tolldiv(\probsize)$.
Now, let $\detrec(\probsize)$ be a solution to the recurrence
$\detrec(\probsize) = \tollwork(\probsize)/\tolldiv(\probsize) +
\detrec(\meanbound(\probsize))$.  If $\tollwork/\tolldiv$ and
$\detrec$ are continuous and $\tollwork/\tolldiv$ is strictly
increasing on the set $\{x\ |\ \tollwork(x)/\tolldiv(x) >0\}$, then
for all positive integers $w$ and $\superadd(\probsize) \geq 1$:
\[
\probgt{\work(\probsize)}{\tolldiv(\probsize)\detrec(\probsize) + w \tollwork(\probsize)} 
  \leq
 \superadd(\probsize)\left(\frac{\meanbound(\probsize)}{\probsize}\right)^w
\]
\end{thm}
Notice that the tail bound has an additional factor of $\superadd(x)$
compared with \thmref{thm:karp-simple}.  Generally we will use
$\superadd(x) = x$, and we will be interested in $w$ such that the
term $(\meanbound(\probsize)/\probsize)^w$ is already exponentially
small compared to $\superadd(x)$, so we will still be able to show
that the desired bound holds with high probability. Note that \thmref{thm:karp-simple} is indeed a special case of these results, since in the degenerate case where $n=1$, we may take $\superadd(\probsize) = \tolldiv(\probsize) = 1$, and we obtain \thmref{thm:karp-simple}.   Moreover,  we prove a more
general form that handles the case when the distribution of the
$\sub_i$ may depend on more than just the problem size.

Our proofs are based on the proof of
\thmref{thm:karp-simple}, so we begin by outlining Karp's proof of
\thmref{thm:karp-simple} in a more general setting than the version
quoted above (\Sref{sec:karp}). We focus only on those parts that will subsequently change in our extensions. Next, we show how to extend the proof
to handle span recurrences (\Sref{sec:span}) to obtain
\thmref{thm:span-simple}. Then, we show that in certain cases,
bounding a work recurrence can be reduced to bounding span recurrences
(\Sref{sec:work}), so that we can deduce \thmref{thm:work-simple} from \thmref{thm:span-simple}. We then apply our method to some examples
(\Sref{sec:applications}). Finally, we discuss related work
(\Sref{sec:conclusion}).
\section{Karp's method}
\label{sec:karp}

In this section we give Karp's proof of a more general version of
\thmref{thm:karp-simple}, where $\work$ and $\sub$ are functions of
problem instances rather than just their sizes. Let $\probDom$ be the
domain of problem inputs, and let $\size : \probDom \to \mathbb{R^+}$
give the size of a problem. Then, we are interested in a recurrence of
the form:
\begin{equation}
\label{eqn:work-rec-unary-general}
\work(\prob) = \tollwork(\size(\prob)) + \work(\sub(\prob))
\end{equation}
where $\sub(\prob)$ is a random variable such that $0 \leq \size(\sub(\prob)) \leq \size(\prob)$. Let us first clarify what we mean by such a probabilistic recurrence. We do not mean that
$\work(\prob)$, as a function from the sample space $\Omega$ to $\mathbb{R}$, is equal to the expression on the right; rather, this is a statement about equality of \emph{distributions}. That is, for all $\prob$ and $r \in \mathbb{R}$:
\begin{equation}
\probgt{\work(\prob)}{r} = \sum_{\probB} \probeq{\sub(\prob)}{\probB} \cdot \probgt{\work(\probB)}{r - \tollwork(\size(\prob))}
\end{equation}
Since we are interested in bounding $\probgt{\work(z)}{r}$, the recurrence only needs to be an inequality:
\begin{equation}
\label{eqn:work-rec-unary-general-cdf}
\probgt{\work(\prob)}{r} \leq \sum_{\probB} \probeq{\sub(\prob)}{\probB} \cdot \probgt{\work(\probB)}{r - \tollwork(\size(\prob))}
\end{equation}
The proof relies on the following lemma:

\begin{lemma}[{\cite[Lemma 3.1]{Karp94}}]
\label{lem:karp}
Let $X$ be a random variable with values in the range $[0,
  x]$. Suppose $f : \mathbb{R} \to \mathbb{R}$ is a
non-negative function such that $f(0) = 0$, and there exists some
constant $c$ such that for all $y \geq c$, $f(y) = 1$ and
 $f(y)/y$ is non-decreasing on the interval $(0,c]$. Then:
\[
  \expected{f(X)} \leq\frac{\expected{X}f(\min(x, c))}{\min(x,c)}
\]
\end{lemma}

\begin{proof}
We start by showing that for all $y$ such that $0 \leq y \leq x$,

\[f(y) \leq \frac{yf(\min(x,c))}{\min(x,c)} \]
If $y = 0$ the result is immediate since $f(0) = 0$. Otherwise, if $x \leq c$, then we have $f(y)/y \leq f(x)/x$, and the result follows. For the case where $ x > c$, we have two sub-cases:

\begin{itemize}
\item If $y \leq c$, then we again have $f(y)/y \leq f(c)/c$ and so we are done.
  
\item If $y > c$, then $\displaystyle f(y) = 1 = f(c) \leq f(c)\frac{y}{c}$.
\end{itemize}
Letting $F$ be the CDF of $X$, we have:

\[
  \expected{f(X)} = \int_{0}^x f(y)\, dF(y) \leq \int_{0}^x \frac{yf(\min(x,c))}{\min(x,c)}\, dF(y)
                  = \frac{\expected{X}f(\min(x,c))}{\min(x,c)}
\]

\end{proof}

\begin{thm}[{\cite[Theorem 3.3]{Karp94}}]
\label{thm:karp}

  Let $\work$ satisfy the recurrence in
  \eqnref{eqn:work-rec-unary-general-cdf} for $\prob$ such that $\size(\prob) > \recterm$. That is, there is some point $\recterm$ at which the recurrence ends.
  Assume $a(x) = 0$ for $x \in [0, \recterm]$ and $a(x)$
  is monotone and continuous on $(\recterm, \infty)$.
  Suppose $\expected{\sub(\probsize)} \leq \meanbound(\probsize)$
  where $0 \leq \meanbound(\probsize) \leq \probsize$ and
  $\meanbound(\probsize)/\probsize$ is non-decreasing. Assume $\detrec :
  \mathbb{R} \to \mathbb{R}$ is a monotone solution to the recurrence
  $\detrec(\probsize) \geq \tollwork(\probsize) +
  \detrec(\meanbound(\probsize))$ which is 
  continuous on $(\recterm, \infty)$. Let $\detrecinv : \mathbb{R} \to
  \mathbb{R^+}$ be an inverse of $\detrec$ above $\recterm$, that is:
  $\detrecinv(\detrec(x)) = x$ for $x > \recterm$, and similarly
  $\detrec(\detrecinv(y)) = y$ for $y > u(\recterm)$. Assume for all
  $x \leq \recterm$, $u(x) = u(\recterm)$, and $u(x)$ bounds the support of $\work(x)$ for $x \leq d$.
  
  Then for all $\prob$ and $r$, $\probgt{\work(\prob)}{r} \leq
  \boundfun_r(\size(\prob))$, where:
  
\begin{enumerate}
\item If $r \leq \detrec(\recterm)$ then $\boundfun_r(x) = 1$
\item If $r > \detrec(\recterm)$:
  \begin{enumerate}
  \item If $x \leq \recterm$ then $\boundfun_r(x) = 0$
  \item If $x > \recterm$ and $\detrec(x) \geq r$ then $\boundfun_r(x) = 1$
  \item If $x > \recterm$ and $\detrec(x) < r$ then
    \[ \boundfun_r(x) = \left(\frac{m(x)}{x}\right)^{\left\lceil \frac{r - \detrec(x)}{\tollwork(x)}\right\rceil}
                 \frac{x}{\detrecinv(r - \tollwork(x)
                           {\left\lceil \frac{r - \detrec(x)}{\tollwork(x)}\right\rceil})} \]
  \end{enumerate}
\end{enumerate}

When $r = \detrec(\probsize) + w\tollwork(\probsize)$ for integer $w$, $\boundfun_r(x)$ simplifies to the bound stated in \thmref{thm:karp-simple} in the introduction.
\end{thm}

\emph{Note:} The statement here is slightly different than Theorem 3.3
in \cite{Karp94}. In particular, there $\detrec$ is
a minimal, invertible solution to the deterministic recurrence and the result uses
the inverse $\detrec^{-1}$ rather than some
$\detrecinv$. Furthermore, $\tollwork$ is assumed to be strictly
increasing on $(\recterm, \infty)$ and continuous everywhere, and
$\meanbound$ is also continuous (so as to imply continuity and invertibility of
$\detrec$). This continuity and strictly increasing assumption rules out
the case where $\tollwork(x) = 1$ for $x > \recterm$ and $\tollwork(x)
= 0$ for $x \leq \recterm$, which arises in some examples, so Karp has
a separate result for this case.  However, the proof does not actually
rely on continuity of $\tollwork$ everywhere or that it be strictly increasing instead of merely monotone, so we prefer the
formulation above which avoids the need for an additional separate theorem.

\begin{proof}

We will only give the proof in the case where for all $\prob$, there
exists some $k$ such that $\probgt{\size(h^k(\prob))}{\recterm} =
0$. See \cite{CSTheoryYoung} for discussion of the proof in a more
general setting.
  
For concision, let us define $K_r(\prob) = \probgt{\work(\prob)}{r}$.
The strategy of the proof is to inductively define a sequence of
functions $\Krec{i}_r$ for $i \in \mathbb{N}$, then show that $K_r(\prob)$ is bounded by $\sup_{i}\Krec{i}_r(\prob)$, and finally prove by induction on $i$ that $\boundfun$ bounds each of the $\Krec{i}$. 

First, \eqnref{eqn:work-rec-unary-general-cdf} implies 
$K_r(\prob) \leq \expected{K_{r-\tollwork(\size(\prob))}(\sub(\prob))}$ when $\size(\prob) > \recterm$. We will use a similar ``recurrence'' to define $\Krec{i}$ by:
\begin{align*}
 \Krec{0}_r(\prob) &=
  \begin{cases} 
    1 & \text{if $r < \detrec(\recterm)$} \\
    0 & \text{otherwise}
  \end{cases} \\
 \Krec{i+1}_r(\prob) &= \expected{\Krec{i}_{r-\tollwork(\size(\prob))}(\sub(\prob))}
\end{align*}
For all $i$, $\Krec{i}_{r}(\prob) \leq \Krec{i+1}_{r}(\prob)$ and
$\Krec{i}_{r}(\prob) \leq 1$ so $\sup_{i}\Krec{i}_{r}(\prob)$
exists. First we note that if $\size(\prob) \leq \recterm$, then
$K_{r}(\prob) \leq \Krec{i}_r(\prob)$.

Next, we claim that if $\Krec{i}_{r}(\prob) < K_{r}(\prob)$ then
$\probgt{\size(\sub^i(\prob))}{\recterm} > 0$. The proof is by
induction on $i$.  If $i = 0$, then we must show
$\probgt{\size(\prob)}{\recterm} > 0$, but this is immediate since if
not we know $K_{r}(\prob) \leq \Krec{0}_{r}(\prob)$, contradicting our assumption of the opposite inequality. For the inductive case, 
$\Krec{i+1}_{r}(\prob) < K_{r}(\prob)$
implies 
$\expected{\Krec{i}_{r - \tollwork(\size(\prob))}(\sub(\prob))}
   < 
  \expected{K_{r-\tollwork(\size(\prob))}(\sub(\prob))}
$.
Hence for some $\probB$ such that $\probeq{\sub(\prob)}{\probB} > 0$, we must have $\Krec{i}_{r - \tollwork(\size(\prob))}(\probB) < K_{r-\tollwork(\size(\prob))}(\probB)$. By the induction hypothesis, we have that $\probgt{\size(\sub^i(\probB))}{\recterm} > 0$, hence $\probgt{\size(\sub^{i+1}(\prob))}{\recterm} \geq {\probeq{\sub(\prob)}{\probB} \cdot \probgt{\size(\sub^i(\probB))}{\recterm}} > 0$. 

This implies that $K_r(\prob) \leq \sup_i \Krec{i}_r(\prob)$, for if
not, then for all $i$ we would have $\Krec{i}_r(\prob) < K_r(\prob)$,
which would imply that for all $i$ there would be some non-zero
probability that the recursion would continue for at least $i$ calls, contradicting our hypothesis about $\work$.

We then show by induction on $i$ that $\Krec{i}_r(\prob) \leq
\boundfun_r(\size(\prob))$. A key part of the induction involves applying
\lemref{lem:karp} to the function $\boundfun_r$ for some particular $r$. To do
so, one must show that $\boundfun_r(x)/x$ is non-decreasing on the
interval $(0, \detrecinv(r)]$. We will not reproduce Karp's proof of
  this fact, since our later proofs will re-use this same
  $\boundfun_r$ function. Nevertheless, the idea is to first divide up $(0, \detrecinv(r)]$
into subintervals on which $\lceil\frac{r-\detrec(x)}{\tollwork(x)}\rceil$ is constant and show that $D_r(x)/x$ is non-decreasing on these intervals. Then, one shows that $D_r(x)/x$ is continuous at the ends of these subintervals, which proves monotonicity on $(0, \detrecinv(r)]$. This is where the continuity assumptions on $\tollwork$ and $\detrec$ are used.
  
From there, the base case of the induction showing $\Krec{i}_r(\prob) \leq \boundfun_r(\size(\prob))$ is straight-forward. For the inductive case, we have
\begin{align*}
\Krec{i+1}_r(\prob) 
 &= \expected{\Krec{i}_{r-\tollwork(\size(\prob))}(\sub(\prob))} \\
 &\leq \expected{\boundfun_{r-\tollwork(\size(\prob))}(\sub(\prob))} \\
 &\leq \expected{\size(\sub(\prob))}
       \frac{\boundfun_{r-\tollwork(\size(\prob))}(\min(\size(\prob),
                                             \detrecinv(r-\tollwork(\size(\prob)))))}
            {\min(\size(\prob), \detrecinv(r-\tollwork(\size(\prob))))} \\
 & \text{(by \lemref{lem:karp} with $X = \size(\sub(\prob))$, 
                                    $f = \boundfun_{r-\tollwork(\size(\prob))}$,
             and $c = \detrecinv(r-\tollwork(\size(\prob)))$)} \\
 &\leq \meanbound(\size(\prob))
       \frac{\boundfun_{r-\tollwork(\size(\prob))}(\min(\size(\prob),
                                  \detrecinv(r-\tollwork(\size(\prob)))))}
            {\min(\size(\prob), \detrecinv(r-\tollwork(\size(\prob))))} \\
 &= \boundfun_r(\size(\prob))
\end{align*}

\end{proof}
\section{Span Recurrences}
\renewcommand{\tollspan}{\tollwork}
\label{sec:span}

We now extend the result from the previous section to handle probabilistic span recurrences of the form:

\begin{equation*}
\label{eqn:span-rec-general}
\spn(\prob) \leq \tollspan(\size(\prob)) + \max_{i=1}^n\spn(\sub_i(\prob))
\end{equation*}
where $0 \leq \size(\sub_i(\prob)) \leq \size(\prob)$. In terms of the CDF of $\spn(\prob)$ this is:
\begin{equation}
\label{eqn:span-rec-general-cdf}
\probgt{\spn(\prob)}{r} \leq \sum_{\probB_1,\dots,\probB_n} \probeq{(\sub_1(\prob), \dots, \sub_n(\prob))}{(\probB_1, \dots, \probB_n)}
                   \cdot \probgt{\max_{i=1}^n\spn(\probB_i)}{r - \tollspan(\size(\prob))}
\end{equation}
Recall from the introduction that the key idea of our result is that we now require a bound on the expected \emph{maximum} of the sizes of the subproblems $\sub_1(\prob), \dots, \sub_n(\prob)$, and we introduce a parameter $\superadd$ such that $\superadd(\size(\probsize)) \geq \sum_{i=1}^n\superadd(\size(\probsizeB_i))$
for all $(\probsizeB_1, \dots, \probsizeB_n)$ in the support of the joint distribution $(\sub_1(\prob), \dots, \sub_n(\prob))$.
We start by proving a generalization of \lemref{lem:karp} which suggests the role that these assumptions will play:

\begin{lemma}
\label{lem:ours}
Let $X: \Omega \to [0, x]^n$ be a random vector where $n \in \mathbb{N^+}$. 
Let $g: \mathbb{R} \to \mathbb{R}^+$ 
be a function such that for all $(y_1, \dots, y_n)$ in the support of $X$, $0 \leq g(y_i)$
and $\sum_i g(y_i) \leq g(x)$.
Suppose $f : \mathbb{R} \to \mathbb{R}$ is a
non-negative function such that $f(0) = 0$, and there exists some
constant $c$ such that for all $y \geq c$, $f(y) = 1$ and
 $f(y)/y$ is non-decreasing on the interval $(0,c]$.
Then:
\[
  \expected{\sum_{i=1}^n g(X_i)f(X_i)} \leq g(x) \frac{\expected{\max_{i=1}^n X_i}f(\min(x, c))}{\min(x,c)}
\]
where $X_i$ is the $i$th component of the random vector $X$.
\end{lemma}

\begin{proof}
It suffices to show that for all $(y_1, \dots, y_n)$ in the support of $X$,
\[\sum_{i=1}^n g(y_i)f(y_i) \leq g(x) \frac{(\max_{i=1}^n y_i)f(\min(x,c))}{\min(x,c)} \]

We have:
\begin{align*}
\sum_{i=1}^n g(y_i)f(y_i) 
  &\leq \left(\max_{j=1}^n f(y_j)\right)\left(\sum_{i=1}^ng(y_i)\right) 
   \leq \left(\max_{j=1}^n f(y_j)\right) g(x)
\end{align*}
Monotonicity of $f(x)/x$ implies $f$ is monotone. Furthermore, by the argument in \lemref{lem:karp}, we know that for each $j$, $f(y_j)\leq \frac{y_j f(\min(x,c))}{\min(x,c)}$. Hence, we have:

\begin{align*}
\sum_{i=1}^n g(y_i)f(y_i) 
  &\leq \left(\max_{j=1}^n f(y_j)\right) g(x) \\
  &\leq f(\max_{j=1}^n y_j) g(x) \\
  &\leq \frac{\left(\max_{j=1}^n y_j\right)f(\min(x,c))}{\min(x,c)} g(x) \\
\end{align*}

\end{proof}

The proof of the next result has the same overall structure as that of \thmref{thm:karp}, except we use the above lemma in place of \lemref{lem:karp}.

\begin{thm}
\label{thm:span}

  Let $\spn$ satisfy the recurrence in
  \eqnref{eqn:span-rec-general-cdf} for $\prob$ such that
  $\size(\prob) > \recterm$. 
  Assume for all $\prob$ there exists $k$ such that for all
  $i_1, \dots, i_k$ we have
  $\probgt{\size(h_{i_1}(\dots(h_{i_k}(\prob))))}{\recterm} = 0$.
  Assume $a(x) = 0$ for $x \in [0, \recterm]$ and $a(x)$ is monotone
  and continuous on $(\recterm, \infty)$.  Suppose
  $\expected{\max_{i=1}^n \sub_i(\probsize)} \leq
  \meanbound(\probsize)$ where $0 \leq \meanbound(\probsize) \leq
  \probsize$ and $\meanbound(\probsize)/\probsize$ is non-decreasing.
  Let $\detrec$, $\detrecinv$ be as in \thmref{thm:karp}, under
  analogous assumptions.  Furthermore, let $\superadd : \mathbb{R}
  \rightarrow \mathbb{R}^+$ be a monotone function such that for all
  $x > \recterm$, $\superadd(x) \geq 1$. Assume for all $\prob$ such
  that $\size(\prob) > \recterm$, and all $(\probB_1, \dots,
  \probB_n)$ in the joint distribution of $(\sub_1(\prob),\dots,
  \sub_n(\prob))$: %
  \[ \sum_{i=1}^n \superadd(\size(\probB_i)) \leq \superadd(\size(\prob)) \]
  Then for all $\prob$ and $r$ such that $\superadd(\size(\prob)) \geq 1 $, we have $\probgt{\spn(\prob)}{r} \leq
  \superadd(\size(\prob))\cdot\boundfun_r(\size(\prob))$ where $\boundfun_r(\size(\prob))$ is as in \thmref{thm:karp}.
\end{thm}

\begin{proof}
Set $K_r(\prob) = \probgt{\spn(\prob)}{r}$. By the union bound, \eqnref{eqn:span-rec-general-cdf} implies 
\[ K_r(\prob) \leq \sum_{i=1}^n \expected{K_{r-\tollspan(\size(\prob))}(h_i(\prob))} \]
As before, we will inductively define a family of functions $\Krec{i}$ using this same recurrence, show that their suprema bounds $K$, and then inductively bound the $\Krec{i}$. Define $\Krec{i}_r(\prob)$ by:
\begin{align*}
 \Krec{0}_r(\prob) &=
  \begin{cases} 
    1 & \text{if $r < \detrec(\recterm)$} \\
    0 & \text{otherwise}
  \end{cases} \\
 \Krec{i+1}_r(\prob) &= \min(1, \sum_{j=1}^n\expected{\Krec{i}_{r-\tollwork(\size(\prob))}(\sub_j(\prob))})
\end{align*}
For all $i$, $\Krec{i}_{r}(\prob) \leq \Krec{i+1}_{r}(\prob)$ and $\Krec{i}_{r}(\prob) \leq 1$, so $\sup_{i}\Krec{i}_{r}(\prob)$ exists. As before, if $\size(\prob) \leq \recterm$, then $K_{r}(\prob) \leq \Krec{i}_r(\prob)$. Here, if $\Krec{n}_{r}(\prob) < K_{r}(\prob)$ then
$\probgt{\size(\sub_{i_1}(\dots(\sub_{i_n}(\prob))))}{\recterm} > 0$ for some sequence $i_1,\dots,i_n$. The proof is by induction on $n$, where the base case is immediate.
For the inductive case, if
$\Krec{n+1}_{r}(\prob) < K_{r}(\prob)$
we must have that $\Krec{n+1}_{r}(\prob) < 1$, so
$\Krec{n+1}_{r}(\prob) = \sum_j \expected{\Krec{n}_{r - \tollwork(\size(\prob))}(\sub_j(\prob))}
   < 
  \sum_j \expected{K_{r-\tollwork(\size(\prob))}(\sub_j(\prob))}
$.
Hence for some $i$ and $\probB$ such that $\probeq{\sub_i(\prob)}{\probB} > 0$, we must have $\Krec{n}_{r - \tollwork(\size(\prob))}(\probB) < K_{r-\tollwork(\size(\prob))}(\probB)$. By the induction hypothesis, $\probgt{\size(\sub_{i_1}(\dots\sub_{i_n}(\probB)))}{\recterm} > 0$. Then the sequence $i_1, \dots, i_n, i$ suffices. It follows that $K_r(\prob) \leq \sup_i \Krec{i}_r(\prob)$, or else our assumption about the termination of the recurrence would be violated.

When $r < \detrec(\recterm)$ and $\size(\prob) \geq \recterm$, we have $\Krec{i}_r(\prob) \leq 1 \leq
\superadd(\size(\prob))\cdot\boundfun_r(\size(\prob))$. For the case when $r \geq \detrec(\recterm)$ we prove by induction on $i$ that $\Krec{i}_r(\prob) \leq
\superadd(\size(\prob))\cdot\boundfun_r(\size(\prob))$. The base case is similar to that of \thmref{thm:karp}. For the inductive step, the only non-trivial sub-case is when $\size(\prob) > \recterm$ and $r > \detrec(\size(\prob))$. Then we have:
\begin{align*}
\Krec{i+1}_r(\prob) 
 &\leq \sum_{j=1}^n \expected{\Krec{i}_{r-\tollwork(\size(\prob))}(\sub_j(\prob))} \\
 &\leq \sum_{j=1}^n \expected{\superadd(\size(\sub_j(\prob)))\cdot
                              \boundfun_{r-\tollwork(\size(\prob))}(\sub_j(\prob))} \\
 &\leq  \superadd(\size(\prob)) \cdot
        \expected{\max_{j=1}^n(\size(\sub_j(\prob)))}
        \frac{\boundfun_{r-\tollwork(\size(\prob))}(\min(\size(\prob), 
                                           \detrecinv(r-\tollwork(\size(\prob)))))}
             {\min(\size(\prob), \detrecinv(r-\tollwork(\size(\prob))))} \\
 & \text{(by \lemref{lem:ours} with $X = (\size(\sub_1(\prob)), \dots, \size(\sub_n(\prob)))$,
               $f = \boundfun_{r-\tollwork(\size(\prob))}$,} \\
 & \text{\qquad $g = \superadd$, and $c = \detrecinv(r-\tollwork(\size(\prob)))$)} \\
 & \leq \superadd(\size(\prob)) \cdot
       \meanbound(\size(\prob))
       \frac{\boundfun_{r-\tollwork(\size(\prob))}(\min(\size(\prob), 
                                     \detrecinv(r-\tollwork(\size(\prob)))))}
            {\min(\size(\prob), \detrecinv(r-\tollwork(\size(\prob))))} \\
 &= \superadd(\size(\prob))\cdot\boundfun_r(\size(\prob))
\end{align*}
In the second line, when we apply the induction hypothesis, we must check that $r - \tollwork(\size(\prob)) \geq \detrec(\recterm)$. This follows from the fact that $r > \detrec(\size(\prob)) \geq \tollwork(\size(\prob)) + \detrec(\meanbound(\size(\prob))) \geq \tollwork(\size(\prob)) + \detrec(\recterm)$.

\end{proof}
\section{Work Recurrences}
\label{sec:work}

Finally we show in this section that in some cases, work-recurrences can be transformed into analogous span-like recurrences. Our starting point is a recurrence:
\begin{equation}
\label{eqn:work-rec-general}
\work(\prob) \leq \tollwork(\size(\prob)) + \sum_{i=1}^n\work(\sub_i(\prob))
\end{equation}
which terminates when $\size(\prob) \leq \recterm$. Let $\tolldiv:
\mathbb{R} \to \mathbb{R}^+$ be a function such that for all
$(\probsizeB_1,\dots,\probsizeB_n)$ in the support of the distribution
$(\size(\sub_1(\prob)), \dots, \size(\sub_n(\prob)))$, $\sum
\tolldiv(\probsizeB_i) \leq \tolldiv(\probsize)$. Consider the random variable $\worktrans(\prob) = \work(\prob)/\tolldiv(\size(\prob))$. We have:

\begin{align*}
\worktrans(\prob) 
  &\leq \frac{\tollwork(\size(\prob))}{\tolldiv(\size(\prob))} 
        + \sum_{i=1}^n \left(\frac{1}{\tolldiv(\size(\prob))}\work(\sub_i(\prob))\right)
  \\
  &= \frac{\tollwork(\size(\prob))}{\tolldiv(\size(\prob))} 
        + \sum_{i=1}^n \left(\frac{\tolldiv(\size(\sub_i(\prob)))}{\tolldiv(\size(\prob))}
                      \worktrans(\sub_i(\prob))\right)
  \\
  &\leq \frac{\tollwork(\size(\prob))}{\tolldiv(\size(\prob))} 
        + \sum_{i=1}^n \left(\frac{\tolldiv(\size(\sub_i(\prob)))}{\tolldiv(\size(\prob))}
                      \max_{j=1}^n\worktrans(\sub_j(\prob))\right)
  \\
  &= \frac{\tollwork(\size(\prob))}{\tolldiv(\size(\prob))} 
        + \sum_{i=1}^n \left(\frac{\tolldiv(\size(\sub_i(\prob)))}{\tolldiv(\size(\prob))}\right)
                      \max_{j=1}^n\worktrans(\sub_j(\prob))
  \\
  &\leq \frac{\tollwork(\size(\prob))}{\tolldiv(\size(\prob))} 
        + \max_{j=1}^n\worktrans(\sub_j(\prob))
\end{align*}

That is, we see $\worktrans$ obeys a span-like recurrence with
toll-function $\tollwork/\tolldiv$. Since
$\probgt{\worktrans(\prob)}{r} =
\probgt{\work(\prob)}{r\tolldiv(\size(\prob))}$, we can apply the
theorem from the previous section to $\worktrans$ and thereby derive a
tail-bound for $\work$, giving us:

\begin{thm}
\label{thm:work}

  Let $\work$ satisfy the recurrence in
  \eqnref{eqn:work-rec-general} for $\prob$ such that $\size(\prob) > \recterm$.
  Assume for all $\prob$ there exists $k$ such that for all
  $i_1, \dots, i_k$ we have
  $\probgt{\size(h_{i_1}(\dots(h_{i_k}(\prob))))}{\recterm} = 0$.
  Suppose $\expected{\max_{i=1}^n \sub_i(\probsize)} \leq \meanbound(\probsize)$
  where $0 \leq \meanbound(\probsize) \leq \probsize$ and
  $\meanbound(\probsize)/\probsize$ is non-decreasing.
  Let $\detrec$ be a solution to the recurrence $\detrec(\probsize) \leq \tollwork(\probsize)/\tolldiv(\probsize) + \detrec(\meanbound(\probsize))$ under analogous assumptions as the previous theorems.
  Furthermore, let $\superadd, \tolldiv : \mathbb{R} \rightarrow \mathbb{R}^+$ be monotone
  functions where for all $x > \recterm$, $\superadd(x) \geq 1$. Assume
  for all $\prob$ such that $\size(\prob) > \recterm$, and all $(\probB_1, \dots, \probB_n)$
  in the joint distribution of $(\sub_1(\prob),\dots, \sub_n(\prob))$:
  \[ \sum_{i=1}^n \superadd(\size(\probB_i)) \leq \superadd(\size(\prob)) \qquad{\text{and}}\qquad
  \sum_{i=1}^n \tolldiv(\size(\probB_i)) \leq \tolldiv(\size(\prob)) \]
  Finally, let $a(x)/g(x) = 0$ for $x \in [0, \recterm]$ and assume $a(x)/g(x)$
  is monotone and continuous on $(\recterm, \infty)$.
  Then for all $\prob$ and $r$ such that $\superadd(\size(\prob)) \geq 1$, we have $\probgt{\work(\prob)}{r} \leq
  \superadd(\size(\prob)) \cdot \boundfun_{r/\tolldiv(\size(x))}(\size(\prob))$ where $\boundfun_r(\size(\prob))$ is as in \thmref{thm:karp}, using the toll function $\tollwork/\tolldiv$.
\end{thm}
\section{Examples}
\label{sec:applications}

We now apply the previous results to examples.

\paragraph{Quicksort} Let $l$ be a list and write $|l|$ for the length of the list. The work and span owing to the comparisons involved in parallel randomized quicksort satisfy the following recurrences:

\begin{align*}
 \work(l) &= (|l| - 1) + \work(\sub_1(l)) + \work(\sub_2(l)) \\
 \spn(l) &= \log |l| + \max(\spn(\sub_1(l)), \spn(\sub_2(l)))
\end{align*}
where $\expected{\max(|\sub_1(l)|, |\sub_2(l)|)} \leq \frac{7|l|}{8}$ when $|l| > 1$ and is $0$ otherwise; moreover, $|\sub_1(l)| + |\sub_2(l)| \leq |l| - 1$. 
Note that this is a case where the distributions of the sizes of the sublists $\sub_1(l)$ and $\sub_2(l)$ depend on more than just the size of $l$: in particular, if $l$ contains multiple versions of the same element, then the sublists will tend to be smaller than if all elements of $l$ are distinct. Therefore, the more general formulations of \thmref{thm:span} and \thmref{thm:work} are useful here.
For span, we apply \thmref{thm:span} with $\meanbound(x) = \frac{7x}{8}$, $\detrec(x) = (\log_{8/7}x + 1)^2$, and $\superadd(x) = x$, to get that for all positive integers $w$,
\[
\probgt{\spn(l)}{(\log_{8/7}|l|+1)^2 + w\log |l|} \leq |l| \left(\frac{7}{8}\right)^w
\]
when $|l| \geq 1$. Take $w$ to be $\lceil \frac{k}{(\log{8/7})}\log_{8/7}|l|\rceil$ for integer $k$ so that we have
\[
\probgt{\spn(l)}{(k+1)(\log_{8/7}|l|+1)^2} \leq \left(\frac{1}{|l|}\right)^{\frac{k}{\log{\frac{8}{7}}} - 1}
\]
For the work, we take $\tolldiv$ to be:
\[ 
  \tolldiv(x) =
    \begin{cases}
       \frac{1}{2} & \text{if $x \leq 1$} \\
       1           & \text{if $1 < x < 2$} \\
       x - 1       & \text{if $x \geq 2$}
    \end{cases}
\]
In this case, in order to solve the recurrence for $\detrec$, it is easier if we take $\meanbound$ to be
\[ 
  \meanbound(x) =
    \begin{cases}
       0           & \text{if $x < 8/7$} \\
       \frac{7x}{8}       & \text{if $x \geq 8/7$}
    \end{cases}
\]
Applying \thmref{thm:work} with $\detrec(x) = \log_{8/7}(x) + 1$ we obtain:
\[
\probgt{\work(l)}{(|l| - 1)(\log_{8/7}|l|+1) + w(|l| - 1)} \leq |l| \left(\frac{7}{8}\right)^w
\]
Here taking $w$ to be $\lceil k \log_{8/7}|l|\rceil$ for integer $k$ we have
\[
\probgt{\work(l)}{(k+1)\left(|l|\log_{8/7}|l|+|l|\right)} \leq  \left(\frac{1}{|l|}\right)^{k - 1}
\]

Karp shows in \citep{Karp94} that the quicksort work recurrence does satisfy the hypotheses of \thmref{thm:karp-work-simple}, which one can check knowing that $\expected{\work(|l|)} = 2((|l|+1)(H_{|l|} - 1) - (|l| - 1))$, where $H_n$ is the $n$th Harmonic number. However, it is much easier to check that $\tolldiv$ satisfies the analogous hypothesis of our result. Moreover, the resulting bound obtained by Karp's \thmref{thm:karp-work-simple},

\[ \probgt{\work(l)}{(a + 1)\expected{\work(l)}} \leq \frac{1}{e^a} \]
is weaker than the bound we obtain above.

\paragraph{Height of search trees}

Consider random binary search trees generated by successively inserting the elements of a permutation of $\{1, \dots, n\}$, with all permutations equally likely. The height $H(n)$ of such a tree obeys the following probabilistic recurrence:
\[
  H(n) = 1 + \max(H(\sub_1(n)), H(\sub_2(n)))
\]
where again $\expected{\max(\sub_1(n), \sub_2(n))} \leq \frac{7n}{8}$ and $\sub_1(n) + \sub_2(n) = n - 1$. Therefore we can apply \thmref{thm:span} with similar choice of $\meanbound$ as in the previous examples, and with $\detrec(n) = \log_{8/7}(n) + 1$ to get:
\[
\probgt{H(n)}{(\log_{8/7}n+1) + w} \leq n \left(\frac{7}{8}\right)^w
\]
Again, for appropriate choice of $w$ this yields:
\[
\probgt{H(n)}{(k+1)(\log_{8/7}n+1)} \leq \left(\frac{1}{n}\right)^{k - 1}
\]
A similar argument works for a large class of tree structures,
such as radix trees and quad trees, so long as we can bound the
expected value of the maximum of the sizes of the subtrees.

\section{Conclusion and Related Work}
\label{sec:conclusion}

We have presented a way to extend \thmref{thm:karp-simple} to
span and work recurrences, answering the open questions posed by Karp
in \cite{Karp94}. Our result for work recurrences is easier to apply
than \thmref{thm:karp-work-simple} and gives stronger bounds in cases
like the analysis of quicksort.

Of course, just as with the original theorems,
more precise bounds can be obtained by instead using more advanced techniques. The books by \citet{FlajoletSedgewick} and \citet{DrmotaTrees}
give an overview of techniques from analytic combinatorics and probability theory that can be used to 
analyze processes that give rise to these kinds of recurrences. 

Nevertheless, Karp has argued convincingly that it is still useful to
have easy to apply cook-book techniques for use in situations where
less precise results are often obtained using ad-hoc methods.  Cook-book
methods are common and well known for analyzing deterministic
divide-and-conquer recurrences (e.g. \cite{AkraB98, Roura01,
  DrmotaS13}).  However, for randomized divide-and-conquer, there are
far fewer results.  The method of \citet{Roura01} applies to the
recurrence for the expected number of comparisons in
quicksort. \citet{BazziM03} extend the Akra--Bazzi
method~\citep{AkraB98} to derive asymptotic bounds for expected work
recurrences.

\citet{ChaudhuriD97} study the unary probabilistic recurrences of
\thmref{thm:karp-simple} and give an alternative result that does not
require $\meanbound(x)/x$ to be non-decreasing (but in this case, the
bound is weaker than \thmref{thm:karp-simple}). They suggest that one advantage of their
proof is that it only uses common tools from probability theory like
Markov's inequality and Chernoff bounds. It would be
interesting to see if some of the ideas we used here
can also be used to extend their version to handle general span and work
recurrences.

In a technical report, \citet{KarpinskiZ91} also modify \thmref{thm:karp-simple}
to try to handle span and work recurrences. However, the bounds they
obtain appear to be weaker than ours. For instance, applied to the span
recurrence of quicksort, their result gives a bound of $\left(1 -
\frac{1}{n}\right)^w$ where our \thmref{thm:span} gives
$n(\frac{7}{8})^w$. That is, the bounds they give do not appear to be
sufficient to establish high probability bounds on run-time. Moreover,
their proof modifies the $D_r(x)$ function and applies
\lemref{lem:karp}, but they do not give a proof for why $D_r(x)/x$
remains non-decreasing, and Karp's strategy for showing that it is
non-decreasing does not appear to generalize to their version.

\paragraph{Acknowledgments} 
The author thanks Robert Harper, Jean-Baptiste Tristan,
Victor Luchangco, Guy Blelloch, and Guy L.~Steele Jr.~for helpful discussions.
This research was conducted with
U.S. Government support under and awarded by DoD, Air Force Office of
Scientific Research, National Defense Science and Engineering Graduate
(NDSEG) Fellowship, 32 CFR 168a.  Any opinions, findings and
conclusions or recommendations expressed in this material are those of
the author and do not necessarily reflect the views of this funding
agency.
 
\bibliographystyle{abbrvnat}
\bibliography{bib}

\begin{thebibliography}{15}
\providecommand{\natexlab}[1]{#1}
\providecommand{\url}[1]{\texttt{#1}}
\expandafter\ifx\csname urlstyle\endcsname\relax
  \providecommand{\doi}[1]{doi: #1}\else
  \providecommand{\doi}{doi: \begingroup \urlstyle{rm}\Url}\fi

\bibitem[Akra and Bazzi(1998)]{AkraB98}
M.~Akra and L.~Bazzi.
\newblock On the solution of linear recurrence equations.
\newblock \emph{Comp. Opt. and Appl.}, 10\penalty0 (2):\penalty0 195--210,
  1998.
\newblock \doi{10.1023/A:1018373005182}.
\newblock URL \url{http://dx.doi.org/10.1023/A:1018373005182}.

\bibitem[Bazzi and Mitter(2003)]{BazziM03}
L.~Bazzi and S.~K. Mitter.
\newblock The solution of linear probabilistic recurrence relations.
\newblock \emph{Algorithmica}, 36\penalty0 (1):\penalty0 41--57, 2003.
\newblock \doi{10.1007/s00453-002-1003-4}.
\newblock URL \url{http://dx.doi.org/10.1007/s00453-002-1003-4}.

\bibitem[Blelloch and Greiner(1995)]{BlellochGr95}
G.~Blelloch and J.~Greiner.
\newblock Parallelism in sequential functional languages.
\newblock In \emph{Proceedings of the 7th International Conference on
  Functional Programming Languages and Computer Architecture}, pages 226--237,
  1995.
\newblock ISBN 0-89791-719-7.

\bibitem[Blumofe and Leiserson(1999)]{BlumofeL99}
R.~D. Blumofe and C.~E. Leiserson.
\newblock Scheduling multithreaded computations by work stealing.
\newblock \emph{J. {ACM}}, 46\penalty0 (5):\penalty0 720--748, 1999.
\newblock \doi{10.1145/324133.324234}.
\newblock URL \url{http://doi.acm.org/10.1145/324133.324234}.

\bibitem[Chaudhuri and Dubhashi(1997)]{ChaudhuriD97}
S.~Chaudhuri and D.~P. Dubhashi.
\newblock Probabilistic recurrence relations revisited.
\newblock \emph{Theor. Comput. Sci.}, 181\penalty0 (1):\penalty0 45--56, 1997.
\newblock \doi{10.1016/S0304-3975(96)00261-7}.
\newblock URL \url{http://dx.doi.org/10.1016/S0304-3975(96)00261-7}.

\bibitem[Cormen et~al.(2009)Cormen, Leiserson, Rivest, and Stein]{CLRS}
T.~H. Cormen, C.~E. Leiserson, R.~L. Rivest, and C.~Stein.
\newblock \emph{Introduction to Algorithms {(3.} ed.)}.
\newblock {MIT} Press, 2009.
\newblock ISBN 978-0-262-03384-8.
\newblock URL \url{http://mitpress.mit.edu/books/introduction-algorithms}.

\bibitem[Devroye(1986)]{Devroye86}
L.~Devroye.
\newblock A note on the height of binary search trees.
\newblock \emph{J. {ACM}}, 33\penalty0 (3):\penalty0 489--498, 1986.
\newblock \doi{10.1145/5925.5930}.
\newblock URL \url{http://doi.acm.org/10.1145/5925.5930}.

\bibitem[Drmota(2001)]{Drmota01}
M.~Drmota.
\newblock An analytic approach to the height of binary search trees.
\newblock \emph{Algorithmica}, 29\penalty0 (1):\penalty0 89--119, 2001.
\newblock \doi{10.1007/BF02679615}.
\newblock URL \url{http://dx.doi.org/10.1007/BF02679615}.

\bibitem[Drmota(2009)]{DrmotaTrees}
M.~Drmota.
\newblock \emph{Random Trees}.
\newblock Springer-Verlag Wien, 2009.
\newblock ISBN 978-3-211-75355-2.

\bibitem[Drmota and Szpankowski(2013)]{DrmotaS13}
M.~Drmota and W.~Szpankowski.
\newblock A master theorem for discrete divide and conquer recurrences.
\newblock \emph{J. {ACM}}, 60\penalty0 (3):\penalty0 16:1--16:49, 2013.
\newblock \doi{10.1145/2487241.2487242}.
\newblock URL \url{http://doi.acm.org/10.1145/2487241.2487242}.

\bibitem[Flajolet and Sedgewick(2009)]{FlajoletSedgewick}
P.~Flajolet and R.~Sedgewick.
\newblock \emph{Analytic Combinatorics}.
\newblock Cambridge University Press, 2009.
\newblock ISBN 978-0-521-89806-5.
\newblock URL
  \url{http://www.cambridge.org/uk/catalogue/catalogue.asp?isbn=9780521898065}.

\bibitem[Karp(1994)]{Karp94}
R.~M. Karp.
\newblock Probabilistic recurrence relations.
\newblock \emph{J. {ACM}}, 41\penalty0 (6):\penalty0 1136--1150, 1994.
\newblock \doi{10.1145/195613.195632}.
\newblock URL \url{http://doi.acm.org/10.1145/195613.195632}.

\bibitem[Karpinski and Zimmermann(1991)]{KarpinskiZ91}
M.~Karpinski and W.~Zimmermann.
\newblock Probabilistic recurrence relations for parallel divide-and-conquer
  algorithms.
\newblock Technical Report TR-91-067, International Computer Science Institute
  (ICSI), 1991.
\newblock URL
  \url{https://www.icsi.berkeley.edu/ftp/global/pub/techreports/1991/tr-91-067.pdf}.

\bibitem[Roura(2001)]{Roura01}
S.~Roura.
\newblock Improved master theorems for divide-and-conquer recurrences.
\newblock \emph{J. {ACM}}, 48\penalty0 (2):\penalty0 170--205, 2001.
\newblock \doi{10.1145/375827.375837}.
\newblock URL \url{http://doi.acm.org/10.1145/375827.375837}.

\bibitem[Young()]{CSTheoryYoung}
N.~Young.
\newblock Answer to: Understanding proof of theorem 3.3 in {Karp's}
  probabilistic recurrence relations.
\newblock Theoretical Computer Science Stack Exchange.
\newblock URL \url{http://cstheory.stackexchange.com/q/37144}.
\newblock URL:http://cstheory.stackexchange.com/q/37144 (version: 2016-12-13).

\end{thebibliography}
\end{document}